\documentclass[12pt,a4paper]{article}
\usepackage{times}
\usepackage{fullpage}
\usepackage{graphicx}
\usepackage{amsthm}
\usepackage{amsmath}
\usepackage{amsfonts}
\usepackage{amssymb}
\usepackage{dsfont}
\usepackage{setspace}
\usepackage{algorithm}
\usepackage{algorithmic}
\usepackage{url}
\usepackage{subfigure}
\usepackage{paralist}
\usepackage{multirow}
\IfFileExists{srcltx.sty} {\usepackage[active]{srcltx}}

\DeclareMathOperator*{\argmin}{arg\,min}
\DeclareMathOperator*{\argmax}{arg\,max}

\newcommand{\abs}[1]{\left\vert#1\right\vert}

\newcommand{\girth}{\mathrm{girth}}%
\newcommand{\wone }{\lVert w \rVert_1}%
\newcommand{\R }{\mathds{R}}%notation for Reals
\newcommand{\N }{\mathds{N}}
\newcommand{\calN }{\mathcal{N}} %neighborhood
\newcommand{\calC }{\mathcal{C}} %code
\newcommand{\calT }{\mathcal{T}} %trees
\newcommand{\calV }{\mathcal{V}} %variable node sets
\newcommand{\calJ }{\mathcal{J}} %constraint node sets

\newcommand{\calB }{\mathcal{B}}

\renewcommand{\Pr }{\mathrm{Pr}}
\newcommand{\reduced}{\mathrm{red}}
\newcommand{\Trim }{\mathrm{Trim}}
\newcommand{\PPrefix }{\mathrm{Prefix}^+}

\newcommand{\cost }{\mathrm{cost}}

\newcommand{\ML}{\textsc{ml}}
\newcommand{\NLO}{\textsc{slo}}
\newcommand{\LO}{\textsc{lo}}
\newcommand{\NWMS}{\textsc{nwms}}

\newcommand{\calCbar }{\overline{\mathcal{C}}}
\newcommand{\calCbarJ}{\overline{\mathcal{C}}^{\mathcal{J}}}

\graphicspath{{figs/}}

\newtheorem{definition}{Definition}
\newtheorem{lemma}[definition]{Lemma}

\newtheorem{proposition}[definition]{Proposition}
\newtheorem{theorem}[definition]{Theorem}
\newtheorem{corollary}[definition]{Corollary}

\newtheorem{example}{Example}

\begin{document}

\title{Hierarchies of Local-Optimality Characterizations in Decoding of Tanner Codes\thanks{A preliminary version of this paper appeared in the proceedings of the IEEE International Symposium on Information Theory, Cambridge, MA,USA, 2012.}}

\author{
      Nissim Halabi \thanks{School of Electrical Engineering, Tel-Aviv University, Tel-Aviv 69978, Israel. \mbox{{E-mail}:\ {\tt nissimh@eng.tau.ac.il}.}}
      \and
      Guy Even \thanks{School of Electrical Engineering, Tel-Aviv University, Tel-Aviv 69978, Israel.  \mbox{{E-mail}:\ {\tt guy@eng.tau.ac.il}.}}}

\date{}

 \maketitle

\begin{abstract}
Recent developments in decoding of Tanner codes with maximum-likelihood certificates are based on a sufficient condition called local-optimality.
We define hierarchies of locally-optimal codewords with respect to two parameters.
One parameter is related to the minimum distance of the local codes in Tanner codes.
The second parameter is related to the finite number of iterations used in iterative decoding.
We show that these hierarchies satisfy inclusion properties as these parameters are increased. In particular, this implies that a codeword that is decoded with a certificate using an iterative decoder after $h$ iterations is decoded with a certificate after $k\cdot h$ iterations, for every integer $k$.
\end{abstract}

%%%%%%%%%%%%%%%%%%%%%%%%%%%%%%%%%%%%%%%%%%%%%%%%%%%%%%%%%%%%%%%%%%%%%%
\section{Introduction}  \label{sec:intro}
%%%%%%%%%%%%%%%%%%%%%%%%%%%%%%%%%%%%%%%%%%%%%%%%%%%%%%%%%%%%%%%%%%%%%%

Local-optimality is often used as a sufficient condition for
successful decoding of finite-length codes (see
e.g.,~\cite{WJW05,ADS09}).  In this work we focus on two parameters of
the local-optimality characterization for Tanner
codes~\cite{EH11}.  The first parameter is related to the minimum
distance of the local codes in (expander) Tanner codes.  The second
parameter is related to the finite number of iterations used in
iterative decoding, even when number of iterations exceeds the girth of the Tanner graph. We define hierarchies of
local-optimality with respect to these parameters. These hierarchies
provide a partial explanation of two questions about successful
decoding with ML-certificates:
(1)~What is the effect of increasing the minimum distance of the local codes in Tanner codes?
(2)~What is the effect of increasing the number of iterations beyond the girth in iterative decoding?

\emph{Previous Work:}
Suboptimal decoding of expander Tanner codes was analyzed in many
works (see e.g.,~\cite{SS96,BZ04,FS05}).  The results in these analyses rely
on: (i)~the expansion properties of the Tanner graph, and
(ii)~constant relative minimum distances of the local codes.  The
error-correcting guarantees in these analyses improve as the expansion factor and relative minimum distance increase.
The first part of our work focuses on the effect of increasing the minimum distance of the local codes on error correcting guarantees of Tanner codes by ML-decoding and LP-decoding.

Density Evolution (DE) is used to study the
asymptotic performance of decoding algorithms based on
Belief-Propagation (BP) (see e.g.,~\cite{RU01,CF02}).  Convergence of
BP-based decoding algorithms to some fixed point was studied
in~\cite{FK00,WF01,WJW05,JP11}.  However, convergence guarantees do not imply
successful decoding after a finite number of iterations.  Korada and
Urbanke~\cite{KU11} provide an asymptotic analysis of iterative
decoding ``beyond'' the girth.  Specifically, they prove that one may
exchange the order of the limits in DE-analysis of BP-decoding under
certain conditions (i.e., variable node degree at least $5$ and
bounded LLRs).  On the other hand, the second part of our work focuses on properties of iterative
decoding of finite-length codes using a finite number of iterations.

A new local-optimality characterization for a codeword in a Tanner
code w.r.t. any MBIOS channel was presented in~\cite{EH11}.  A
locally-optimal codeword is guaranteed to be both the unique
maximum-likelihood (ML) codeword as well as the unique LP-decoding
codeword. The characterization of local-optimality for Tanner codes
has three parameters: (i)~a height $h\in\N$, (ii)~level weights
$w\in\R_+^h$, and (iii)~a degree $2\leqslant d\leqslant d^*$, where
$d^*$ is the smallest minimum distance of the component local codes.

A new message-passing decoding algorithm, called \emph{normalized
  weighted min-sum} (\NWMS), was presented for Tanner codes with single parity-check (SPC)
local codes~\cite{EH11}.  The \NWMS\ decoder is guaranteed to compute the
ML-codeword in $h$ iterations provided that a locally-optimal codeword
with height parameter $h$ exists. The number of iterations $h$ may exceed
the girth of the Tanner graph.

\emph{Contribution:}
To obtain one of the hierarchy results, we needed a new definition of local-optimality called
\emph{strong local-optimality}. We prove that if a codeword is strongly
locally-optimal, then it is also locally-optimal (Lemma~\ref{lemma:NLOinLO}). Hence, previous results
proved for local-optimality~\cite{EH11} hold also for strong
local-optimality.

We present two combinatorial hierarchies: (1)~A \emph{hierarchy of
  local-optimality based on degrees}. The degree hierarchy states that a locally-optimal codeword $x$ with degree parameter $d$ is also locally-optimal
with respect to any degree parameter $d'>d$. The degree hierarchy
implies that the occurrence of local-optimality does not decrease as
the degree parameter increases.
(2)~A \emph{hierarchy of strong local-optimality based on height}. The
height hierarchy states that if a codeword $x$ is
strongly locally-optimal with respect to height parameter $h$, then it
is also strongly locally-optimal with respect to every height parameter
that is an integer multiple of $h$. The height hierarchy proves, for
example, that the performance of iterative decoding with an
ML-certificate (e.g., \NWMS) of finite-length Tanner
codes with SPC local codes does not degrade as the number of
iterations grows, even beyond the girth of the Tanner graph.

\paragraph{Organization.} In Section~\ref{sec:trim} we introduce a key trimming procedure used in the proofs of the hierarchies. In Section~\ref{sec:degreeHier} we prove that the degree-based hierarchy induces a chain of inclusions of locally-optimal codewords and LLRs. In Section~\ref{sec:heightHier} we prove a height-based hierarchy over strong local-optimality. We show that strong local-optimality implies local-optimality. Numerical results of strong local-optimality and local-optimality with respect to the height hierarchy are presented in Section~\ref{sec:numerical}.
We conclude with a discussion in Section~\ref{sec:discussion}.

%%%%%%%%%%%%%%%%%%%%%%%%%%%%%%%%%%%%%%%%%%%%%%%%%%%%%%%%%%%%%%%%%%%%%
\section{Preliminaries}
\label{sec:prelim}
%%%%%%%%%%%%%%%%%%%%%%%%%%%%%%%%%%%%%%%%%%%%%%%%%%%%%%%%%%%%%%%%%%%%%
%\emph{Graph Terminology:}
\paragraph{Graph Terminology.}
Let $G=(V,E)$ denote an undirected graph. Let $\mathcal{N}_G(v)$ denote the set
of neighbors of node $v \in V$, and let $\deg_G(v)\triangleq\lvert\mathcal{N}_G(v)\rvert$ denote the degree of node $v$ in graph $G$.
A path
$p=(v,\ldots,u)$ in $G$ is a sequence of vertices such that there
exists an edge between every two consecutive nodes in the sequence
$p$. A path $p$ is \emph{backtrackless} if every three consecutive vertices along $p$ are distinct (i.e., a subpath $(u,v,u)$ is not allowed).
Let $\lvert p\rvert$ denote the number of edges in $p$.
Let $\girth(G)$ denote the length of the shortest cycle in $G$.
Given a graph G, an \emph{edge-labeling} is a function that maps edges of G to a set of labels. In this case, $G$ is called an \emph{edge-labeled graph}.

%\emph{Tanner-codes and ML-decoding:}
\paragraph{Tanner-codes.}
Let $G=(\calV\cup\calJ,E)$ denote an edge-labeled
bipartite-graph, where $\calV=\{v_1,\ldots,v_N\}$ is a set of
$N$ vertices called \emph{variable nodes}, and $\calJ=\{C_1,\ldots,C_J\}$ is a set of $J$ vertices called \emph{local-code nodes}.
The edge labeling is specified by an ordering $1,\ldots,\deg_G(C_j)$ to edges incident to each local-code node $C_j$, and hence specifies an order on $\mathcal{N}_G(C_j)$ with respect to $C_j$ for every $1\leqslant j\leqslant J$.
We associate with each local-code node $C_j$ a linear code $\calCbar^j$ of length $\deg_G(C_j)$. Let $\calCbarJ \triangleq \big\{\calCbar^j\ :\ 1\leqslant j \leqslant J \big\}$ denote the set of \emph{local codes}, one for each local code node. We say that $v_i$ \emph{participates} in $\calCbar^j$ if $(v_i,C_j)$ is an edge in $E$.

A word $x=(x_1,\ldots,x_N)\in\{0,1\}^N$ is an assignment to variable nodes in $\calV$ where $x_i$ is assigned to $v_i$.
The \emph{Tanner code} $\calC(G,\calCbarJ)$ based on the labeled
\emph{Tanner graph} $G$ is the set of vectors $x\in\{0,1\}^N$
such that the projection of $x$ onto
entries associated with $\calN_G(C_j)$ is a codeword in $\calCbar^j$ for every $j \in \{1,\ldots,J\}$.
Let $d_j$ denote the minimum distance of the local code $\calCbar^j$.
The \emph{minimum local distance} $d^*$ of a Tanner code
$\calC(G,\calCbarJ)$ is defined by $d^*\triangleq\min_j d_j$. We assume that $d^*\geq 2$.

If the bipartite graph is $(d_L,d_R)$-regular, then the graph defines a \emph{$(d_L,d_R)$-regular Tanner code}.
If the Tanner graph is sparse, i.e., $|E|=O(N)$, then it defines a
\emph{generalized low-density parity-check (GLDPC)} code. Tanner codes with
single parity-check (SPC) local codes that are based on sparse Tanner graphs are called \emph{low-density parity-check (LDPC) codes}.

%\emph{Communicating over memoryless channels:}
\paragraph{Communicating over memoryless channels.}
Let $c_i\in\{0,1\}$ denote the $i$th transmitted binary symbol (channel input), and let $y_i\in\R$ denote the $i$th received symbol (channel output).
A \emph{memoryless binary-input output-symmetric} (MBIOS) channel is defined by a conditional probability density function $f(y_i|c_i=a)$ for $a\in\{0,1\}$, that satisfies $f(y_i|0) = f(-y_i|1)$.
In MBIOS channels, the \emph{log-likelihood ratio} (LLR) vector $\lambda \in \R^N$ is defined by $\lambda_i (y_i) \triangleq
\ln\big(\frac{f(y_i|c_i=0)}{f(y_i|c_i=1)}\big)$ for every
input bit $i$. For a code $\calC$, \emph{Maximum-Likelihood (ML)
decoding} is equivalent to finding a word $\hat{x}^{\ML}$ that satisfies $\hat{x}^{\ML}(y) = \argmin_{x \in \calC} \langle\lambda(y) , x \rangle$.

%\emph{Local-Optimality Characterization:}
%\paragraph{Local-Optimality Characterization.}
\paragraph{Deviations.}
A new characterization for local-optimality of Tanner codes was presented in~\cite{EH11} as extension to~\cite{ADS09, Von10}.
Local-optimality is a combinatorial characterization of a codeword with respect to a given LLR vector. This characterization of local optimality is based on a set of vectors, called deviations, induced by combinatorial structures in computation trees of the Tanner graph. The set of deviations is specified in~(\ref{eqn:deviations}), and local-optimality is defined in Definition~\ref{def:localOptimality}.
We present a few definitions, examples of which appear in Example~\ref{exmple:deviation}

\begin{definition}[Path-Prefix Tree]\label{def:ppt}
Consider a graph $G=(V,E)$ and a node $r\in V$. Let $\hat{V}$ denote the set of all backtrackless paths in $G$ with length at most $h$ that start at node $r$, and let
$\hat{E} \triangleq \big\{(p_1,p_2)\in\hat{V}\times\hat{V}\ \big\vert\ p_1\ \mathrm{is~a~prefix~of~}p_2,~\lvert p_1\rvert+1=\lvert p_2\rvert \big\}$.
We denote the zero-length path in $\hat{V}$ by $(r)$.
The directed graph $(\hat{V},\hat{E})$ is called the \emph{path-prefix tree} of $G$ rooted at node $r$ with height $h$, and is denoted by $\calT_r^{h}(G)$.
\end{definition}
\noindent The graph $\calT_r^{h}(G)$ is obviously acyclic and is an out-tree rooted at $(r)$. Path-prefix trees of $G$ that are rooted in variable nodes are often called \emph{computation trees}.

We use the following notation. Vertices in $\calT_r^{h}(G)$
are paths in $G$, and are denoted by $p$ and
$q$, while vertices in $G$ are denoted by $u,v,r$. For a path
$p\in\hat{V}$, let $t(p)$ denote the last vertex (\emph{target}) of path
$p$. Denote by $\PPrefix(p)$ the set of proper prefixes of the path
$p$, i.e., not including the root and $p$. Formally,
\begin{equation*}
\PPrefix(p) = \big\{q\ \big\vert\ q \mathrm{\ is\ a\ prefix\ of\ }p,\ 1 \leqslant\rvert q \lvert<\lvert p \rvert\big\}.
\end{equation*}
When $G=(\calV\cup\calJ,E)$ is a Tanner graph, let $\hat{\calV}$ denote the set of paths in $\hat{V}$ that end in a variable node, i.e., $\hat{\calV}\triangleq \{p\ \vert\ p\in\hat{V},\ t(p)\in\calV\}$. Let $\hat{\calJ}$ denote the set of paths in $\hat{V}$ that end in a local-code node, i.e., $\hat{\calJ}\triangleq \{p\ \vert\ p\in\hat{V},\ t(p)\in\calJ\}$.
Paths in $\hat{\calV}$ are called \emph{variable paths}, and paths in  $\hat{\calJ}$ are called \emph{local-code paths}.

\begin{definition}[$d$-tree]\label{def:Dtree}
Let $G=(\calV\cup\calJ,E)$ denote a Tanner graph.
A subtree $\calT \subseteq \calT_r^{2h}(G)$ is a \emph{$d$-tree} if:
\begin{inparaenum}[(i)]
\item $r$ is a variable node,
\item $\calT$ is rooted at the root $(r)$ of $\calT_r^{2h}$,
\item for every local-code path $p\in\calT\cap\hat\calJ$, $\deg_\calT(p)=d$, and
\item for every variable path $p\in\calT\cap\hat\calV$, $\deg_\calT(p)=\deg_{\calT_r^{2h}}(p)$.
\end{inparaenum}
\end{definition}
\noindent Let $\calT[r,2h,d](G)$ denote the set of all $d$-trees rooted at $r$ that are subtrees of $\calT_r^{2h}(G)$.

\begin{definition} [$w$-weighted subtree]\label{def:weightedSubtree}
  Let $\calT = (\hat{\calV}\cup\hat{\calJ},\hat{E})$ denote a subtree
  of $\calT_r^{2h}(G)$, and let $w=(w_1,\ldots,w_h)\in\R_+^h\setminus\{0^h\}$ denote a
  non-negative weight vector.  Let
  $w_\calT:\hat{\mathcal{V}}\rightarrow \R$ denote the weight function defined as follows. If $p$ is a zero-length variable path, then $w_\calT(p)=0$. Otherwise,
\begin{equation}\label{eqn:w-weights}
   w_\calT(p) \triangleq \frac{w_\ell}{
\wone}\cdot\frac{1}{\deg_G\big(t(p)\big)}\cdot \prod_{q\in \PPrefix(p)}\frac{1}{\deg_{\calT}(q)-1},
\end{equation}
where $\ell = \big\lceil\frac{\lvert p\rvert}{2}\big\rceil$. We refer to $w_\calT$ as a $w$-weighted subtree.
\end{definition}

For any $w$-weighted subtree $w_\calT$ of $\calT_{r}^{2h}(G)$, let $\pi_{G,\calT,w}:\calV \rightarrow \R$ denote a function whose values correspond to the projection of $w_\calT$ on the Tanner graph $G$. That is, for every variable node $v$ in $G$,
\begin{equation}
\pi_{G,\calT,w}(v) \triangleq \sum_{\{p\in\calT\mid t(p)=v\}}w_\calT(p).
\end{equation}

For a Tanner code $\calC(G)$, let $\mathcal{B}_d^{(w)}\subseteq
[0,1]^N$ denote the set of all projections of
$w$-weighted $d$-trees on $G$. That is,
\begin{equation}\label{eqn:deviations}
\mathcal{B}_d^{(w)} \triangleq \big\{\pi_{G,\calT,w}\ \big\vert\ \calT\in\bigcup_{r\in\calV}\calT[r,2h,d](G)\big\}.
\end{equation}
Vectors in $\mathcal{B}_d^{(w)}$ are called \emph{deviations}.

\begin{example} [deviation induced by a normalized weighted subtree in computation tree of the Tanner graph]\label{exmple:deviation}
Figure~\ref{fig:DeviationExample} depicts a construction of a $3$-tree as a subtree of a path-prefix tree with height $4$ of a Tanner graph. The Tanner graph illustrated in Figure~\ref{fig:tannerGraph} contains $4$ variable nodes (depicted by circles) and $3$ local-code nodes (depicted by squares). We label the variable nodes by `$a$',`$b$',`$c$', and `$d$', and the local-code nodes by `$X$',`$Y$', and `$Z$'. Figure~\ref{fig:ppt} depicts the path-prefix tree of $G$ rooted at variable node `$b$' with height $4$, denoted by $\calT_b^4(G)$. The nodes of $\calT_b^4(G)$ correspond to backtrackless paths in $G$. We depict, for example, the variable paths $(b)$, $(b,Y,c)$, and $(b,Y,c,Z,a)$ and the local-code paths $(b,Y)$ and $(b,Y,c,Z)$. Figure~\ref{fig:3tree} depict a $3$-tree in $\calT_b^4(G)$. Denote this $3$-tree by $\calT$. The degree of every variable path in $\calT$ equals to its degree in the path-prefix tree $\calT_b^4(G)$, and the degree of every local-code path in $\calT$ equals exactly $3$. We depict every variable path in the path-prefix tree that ends at node `$a$' by a filed circle. Every other path node $q\in\calT$ is labeled within the node by $t(q)$, i.e., the last node in the path $q$.

Let $w=(2,4)\in\R^2$. The weight function of the $w$-weighted $3$-tree $\calT$ for variable path $p:=(b,Y,c,Z,a)$ is calculated as follows. Note that $\lvert p\rvert=4$, $t(p)=a$, and $\PPrefix\big((b,Y,c,Z,a)\big)=\{(b,Y), (b,Y,c), (b,Y,c,Z)\}$. Then,
\begin{align*}
w_\calT(p) &= \frac{w_2}{
\wone}\cdot\frac{1}{\deg_G(a)}\cdot \prod_{q\in \PPrefix(p)}\frac{1}{\deg_{\calT}(q)-1}\\
&= \frac{4}{2+4}\cdot\frac{1}{3}\cdot\frac{1}{2\cdot2\cdot2} = \frac{1}{36}.
\end{align*}
Similarly, $w_\calT\big((b,Z,a)\big)=\frac{1}{18}$.

The projection of $w_\calT$ on the Tanner graph $G$ for variable node $a$ is calculated by summing up all the weights of the variable paths in $\calT$ that end at $a$. For $\calT$ depicted in Figure~\ref{fig:3tree}, $\pi_{G,\calT,w}(a)=\frac{1}{18}+\frac{1}{36}+\frac{1}{36}+\frac{1}{36}+\frac{1}{36}+\frac{1}{36}=\frac{7}{36}$, The deviation that corresponds to $\calT$ is $\beta=\pi_{G,\calT,w}=(\frac{7}{36},\frac{4}{18},\frac{1}{4},\frac{11}{36})\in\R^4$.

\begin{figure}[]
\centering
\subfigure[]{\includegraphics[width=0.15\textwidth]{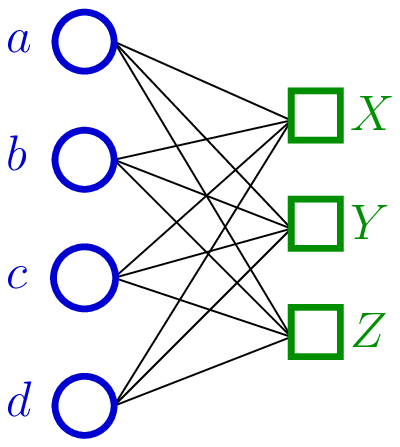}
\label{fig:tannerGraph}}\\
\subfigure[]{\includegraphics[width=0.4\textwidth]{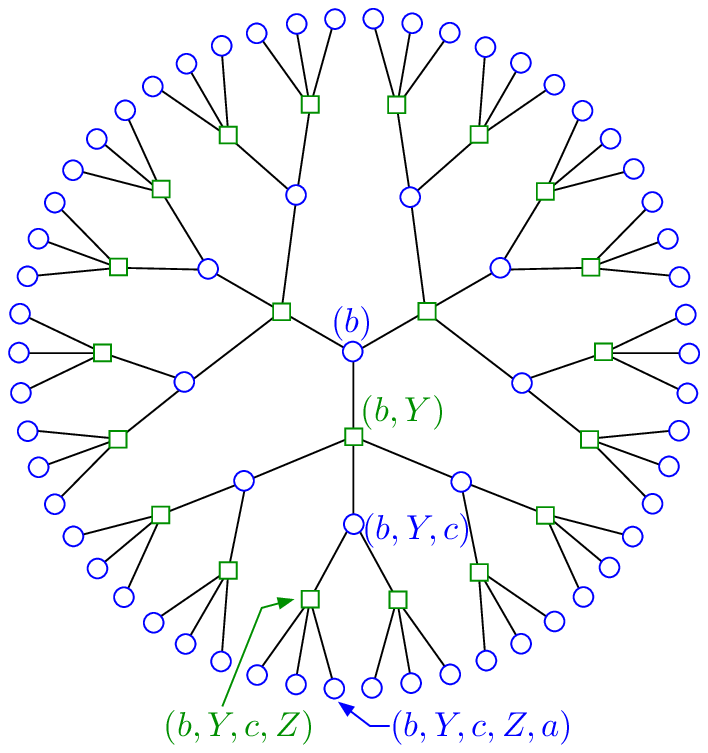}
\label{fig:ppt}}
\subfigure[]{\includegraphics[width=0.6\textwidth]{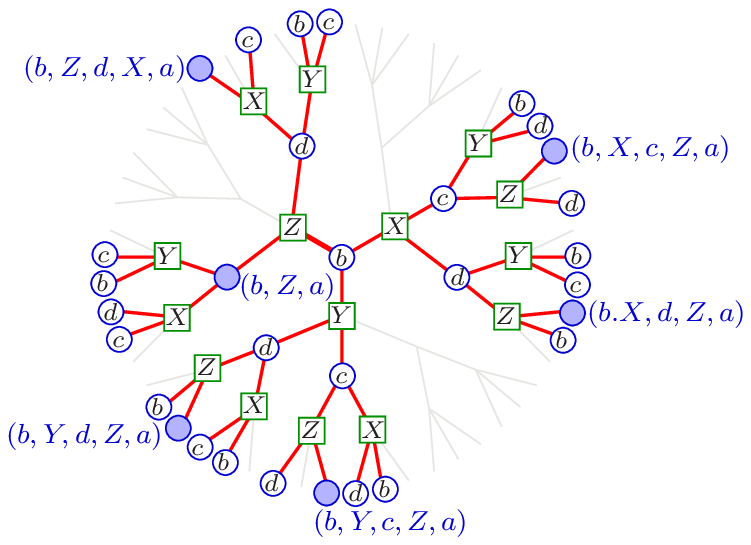}
\label{fig:3tree}}
\label{fig:DeviationExample}
\caption{Example of a $3$-tree: subtree in a computation tree of the Tanner graph.
\newline{\small (a) Tanner graph $G$. Variable nodes marked by circles and labeled by `$a$',`$b$',`$c$',`$d$'. Local-codes nodes marked by squares and labeled by `$X$',`$Y$',`$Z$'. (b) The Path-prefix tree (computation tree) of Tanner graph $G$ rooted at variable node `$b$' with height $4$. (c) A $3$-tree (d=3).
Consider a variable node $a\in\calV$. Each node $p$ in the $3$-tree that is a variable-path that ends in the variable node `$a$' (i.e., the path $p$ ends in the variable node `$a$' of $G$) is depicted by a filled circle, and the path it represents is written next to it. Other nodes (both variable paths and local-code paths) are labeled by their last node.}}
\end{figure}
\end{example}

\paragraph{Local-Optimality Characterization.}
For two vectors $x \in \{0,1\}^N$ and $f \in [0,1]^N$, let $x\oplus f \in [0,1]^N$ denote the \emph{relative point} defined by $(x\oplus f)_i \triangleq |x_i-f_i|$~\cite{Fel03}.

\begin{definition}[local-optimality, \cite{EH11}] \label{def:localOptimality}
A codeword $x \in \calC(G)$ is
  \emph{$(h,w,d)$-locally optimal with respect to $\lambda \in \R^N$}
  if for all vectors $\beta \in \mathcal{B}_d^{(w)}$,
\begin{equation}
\langle \lambda,x \oplus \beta \rangle > \langle \lambda, x \rangle.
\end{equation}
\end{definition}

The following theorem states a combinatorial condition that is sufficient for both ML-optimality and LP-optimality given a channel observation.
\begin{theorem}[local-optimality is sufficient for ML and LP, \cite{EH11}]\label{thm:MLsufficient}
Let $\lambda\in\R^N$ denote the LLR vector received from the channel. If $x$ is an $(h,w,d)$-locally optimal codeword w.r.t. $\lambda$ and some $2\leqslant d \leqslant d^{*}$, then (1)~$x$ is the unique maximum-likelihood codeword w.r.t. $\lambda$, and (2)~$x$ is the unique optimal solution of the LP-decoder given $\lambda$.
\end{theorem}

For a word $x\in\{0,1\}^N$, let $(-1)^x\in\{\pm1\}^N$ denote a vector whose $i$th component equals $(-1)^{x_i}$. Denote by $0^N$ the all-zero vector of length $N$.
For two vectors $y,z\in \R^N$, let ``$\ast$'' denote coordinatewise
multiplication, i.e., $y\ast z \triangleq (y_1\cdot z_1,\ldots,
y_N\cdot z_N)$.
\begin{proposition}[\cite{EH11}]\label{proposition:isoLO}
For every $\lambda\in\R^N$ and every $\beta\in[0,1]^N$,
\[\langle (-1)^x\ast\lambda,\beta\rangle  = \langle  \lambda,x\oplus\beta\rangle-\langle\lambda,x\rangle.\]
\end{proposition}
The following proposition states that the mapping
$(x,\lambda)\mapsto(0^N,(-1)^x\ast\lambda)$ preserves local-optimality.
\begin{proposition}[symmetry of local-optimality, \cite{EH11}]\label{proposition:LOsymmetry}
  For every $x\in\calC$, $x$ is $(h,w,d)$-locally optimal w.r.t.
  $\lambda$ if and only if $0^N$ is $(h,w,d)$-locally optimal w.r.t.
  $(-1)^x\ast\lambda$.
\end{proposition}

%%%%%%%%%%%%%%%%%%%%%%%%%%%%%%%%%%%%%%%%%%%%%%%%%%%%%%%%%%%%%%%%%%%%%
\section{Trimming Subtrees from a Path-Prefix Tree}
\label{sec:trim}
%%%%%%%%%%%%%%%%%%%%%%%%%%%%%%%%%%%%%%%%%%%%%%%%%%%%%%%%%%%%%%%%%%%%%

Let $\calT_q$ denote the subtree of a path-prefix tree $\calT$ hanging from path $q$, i.e., the subtree induced by $\hat{V}_q\triangleq\{p\in\hat\calV\cup\hat\calJ\mid q\in\PPrefix(p)\ \mathrm{or}\ p=q\}$ (see Figure~\ref{fig:trimProof}).
Let $\Trim(\calT,q)$ denote the subtree of $\calT$ obtained by deleting the subtree $\calT_q$ from $\calT$. Formally, $\Trim(\calT,q)$ is the path-prefix subtree of $\calT$ induced by $\hat\calV\cup\hat\calJ\setminus\hat{V}_q$.
Note that if $q'$ is a sibling of $q$ (i.e., $q'$ differs from $q$ only in the last edge), then the degree of the parent of $q$ and $q'$ decreases by one as a result of trimming $\hat{V}_q$. Hence, $w_{\calT}(q'')<w_{\Trim(\calT,q)}(q'')$ for every variable path $q''\in\hat{V}_{q'}$.
\begin{figure}
  \begin{center}
 \includegraphics[width=0.3\textwidth]{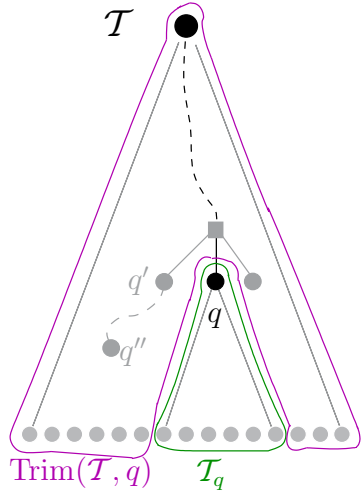}
 \caption{Trimmed tree of $\calT$ induced by $q$.}
  \label{fig:trimProof}
  \end{center}
\end{figure}

The proofs of the hierarchies presented in the following sections are based on the following lemma.
\begin{lemma}\label{lemma:trim}
 Let $\calT$ denote a subtree of a path-prefix tree $\calT_{r}^{2h}(G)$. For every path $p\in\calT$ with at least two children in $\calT$, there exists at least one child $p'$ of $p$, such that
\[\langle\lambda,\pi_{G,\calT,w}\rangle \geqslant \langle\lambda,\pi_{G,\Trim(\calT,p'),w}\rangle.\]
\end{lemma}
\begin{proof}
See Appendix~\ref{app:trimProof}.
\end{proof}

%%%%%%%%%%%%%%%%%%%%%%%%%%%%%%%%%%%%%%%%%%%%%%%%%%%%%%%%%%%%%%%%%%%%%%
\section{Degree Hierarchy of Local-Optimality}
\label{sec:degreeHier}
%%%%%%%%%%%%%%%%%%%%%%%%%%%%%%%%%%%%%%%%%%%%%%%%%%%%%%%%%%%%%%%%%%%%%%

Let $\Lambda\subseteq\R^N$ denote a set of LLR vectors.
Denote by $\LO_{\calC,\Lambda}(h,w,d)$ the set of pairs $(x,\lambda)\in\calC\times\Lambda$ such that $x$ is $(h,w,d)$-locally optimal w.r.t. $\lambda$. Formally,
\begin{equation}
\LO_{\calC,\Lambda}(h,w,d)\triangleq
\big\{(x,\lambda)\in\calC\times\Lambda\mid x \mathrm{\ is}\ (h,w,d)\mathrm{-locally\ optimal\ w.r.t.}\ \lambda\big\}.
\end{equation}

The following theorem derives an hierarchy on the ``density'' of deviations in local-optimality characterization.
\begin{theorem}[$d$-Hierarchy of local-optimality]\label{thm:d-hierarchy}
Let $2 \leqslant d
  < d^*$. For every $\Lambda\subseteq\R^N$,
  \[\LO_{\calC,\Lambda}(h,w,d)\subseteq \LO_{\calC,\Lambda}(h,w,d+1).\]
\end{theorem}
\begin{proof}
We prove the contrapositive statement. Assume that $x$ is not $(h,w,d+1)$-locally optimal w.r.t. $\lambda$. By Proposition~\ref{proposition:LOsymmetry}, $0^N$ is not $(h,w,d+1)$-locally optimal w.r.t. $\lambda^0\triangleq (-1)^x\ast\lambda$. Hence, there exists a deviation $\beta=\pi_{G,\calT,w}\in\calB_d^{(w)}$ such that $\langle\lambda^0,\beta\rangle\leqslant0$. Let $\calT$ denote the $(d+1)$-tree that corresponds to the deviation $\beta$.

Consider the following iterative trimming process.
Start with the $(d+1)$-tree $\calT$ and let $\calT\gets\calT'$; While there exists a local-code path $p\in\calT'$ such that $\deg_{\calT'}(p)=d+1$ do: $\calT' \gets \Trim(\calT',q)$ where $q$ is a child of $p$ such that $\langle\lambda^0,\pi_{G,\calT',w}\rangle \geqslant \langle\lambda^0,\pi_{G,\Trim(\calT',q),w}\rangle$.

Lemma~\ref{lemma:trim} guarantees that the iterative trimming process halts with a $d$-tree $\calT'$
whose corresponding deviation $\beta'=\pi_{G,\calT',w}$ satisfies $\langle \lambda^0,\beta'\rangle\leqslant\langle\lambda^0,\beta\rangle\leqslant0$.
We conclude by Proposition~\ref{proposition:LOsymmetry} that $x$ is not $(h,w,d)$-locally optimal w.r.t. $\lambda$, as required.
\end{proof}

\medskip \noindent We conclude that for every $2 \leqslant d
  < d^*$,
\begin{align*}
\Pr_\lambda\big\{x\mathrm{\ is\ }(h,w,d+1)\mathrm{-locally\ optimal\ w.r.t.\ }&\lambda\big\} \geqslant\\
&\Pr_\lambda\big\{x\mathrm{\ is\ }(h,w,d)\mathrm{-locally\ optimal\ w.r.t.\ }\lambda\big\}.
\end{align*}

%%%%%%%%%%%%%%%%%%%%%%%%%%%%%%%%%%%%%%%%%%%%%%%%%%%%%%%%%%%%%%%%%%%%%%
\section{Height Hierarchy of Strong Local-Optimality}
\label{sec:heightHier}
%%%%%%%%%%%%%%%%%%%%%%%%%%%%%%%%%%%%%%%%%%%%%%%%%%%%%%%%%%%%%%%%%%%%%%

In this section we introduce a new combinatorial characterization named \emph{strong local-optimality}.
We prove that if a codeword is strongly locally-optimal then it is also locally-optimal. The other direction is not true in general. We prove a hierarchy on strong local-optimality based on the height parameter. We discuss in Section~\ref{sec:discussion} on the implications of the height hierarchy on iterative message-passing decoding of Tanner codes.

\begin{definition}[reduced $d$-tree]\label{def:reducedDtree}
Denote by $\calT_r^{2h}(G)=(\hat\calV\cup\hat\calJ,\hat{E})$ the path-prefix tree of a Tanner graph $G$ rooted at node $r\in\calV$. A subtree $\calT \subseteq \calT_r^{2h}(G)$ is a \emph{reduced $d$-tree} if:
\begin{inparaenum}[(i)]
\item $\calT$ is rooted at $r$,
\item $\deg_\calT\big((r)\big)=\deg_G(r)-1$,
\item for every local-code path $p\in\calT\cap\hat\calJ$, $\deg_\calT(p)=d$, and
\item for every non-empty variable path $p\in\calT\cap\hat\calV$, $\deg_\calT(p)=\deg_{\calT_r^{2h}}(p)$.
\end{inparaenum}
\end{definition}
The only difference between Definition~\ref{def:Dtree} ($d$-tree) to a reduced $d$-tree is that the degree of the root in a reduced $d$-tree is smaller by 1 (as if the root itself hangs from an edge)\footnote{This difference is analogous to the ``edge'' versus ``node'' perspectives of tree ensembles in the book Modern Coding Theory~\cite{RU08}}.

Let $\calT^\reduced[r,2h,d](G)$ denote the set of all reduced $d$-trees rooted at $r$ that are subtrees of $\calT_r^{2h}(G)$. For a Tanner code $\calC(G)$, let $\overline{\calB}^{(w)}_d\subseteq[0,1]^N$ denote the set of all projections of $w$-weighted reduced $d$-trees on $G$. That is,
\begin{equation}\overline{\calB}^{(w)}_d \triangleq \big\{\pi_{G,\calT,w}\big\vert\calT\in\bigcup_{r\in\calV}\calT^\reduced[r,2h,d](G)\big\}.
\end{equation}
Vectors in $\overline{\calB}^{(w)}_d$ are referred to as \emph{reduced deviations}.

The following definition is analogues to Definition~\ref{def:localOptimality} (local-optimality) using reduced deviations instead of deviations.
\begin{definition}[strong local-optimality] \label{def:nealyLO}
Let $\mathcal{C}(G) \subset \{0,1\}^N$ denote a Tanner code. Let $w \in \R_+^h\backslash\{0^h\}$ denote a
  non-negative weight vector of length $h$ and let $d \geqslant 2$.
  A codeword $x \in \calC(G)$ is
  \emph{$(h,w,d)$-strong locally-optimal with respect to $\lambda \in \R^N$}
  if for all vectors $\beta \in \overline{\mathcal{B}}_d^{(w)}$,
\begin{equation}
\langle \lambda,x \oplus \beta \rangle > \langle \lambda, x \rangle.
\end{equation}
\end{definition}

Denote by $\NLO_{\calC,\Lambda}(h,w,d)$ the set pairs $(x,\lambda)\in\calC\times\Lambda$ such that $x$ is $(h,w,d)$-strong locally-optimal w.r.t. $\lambda$. Formally,
\begin{equation}
\NLO_{\calC,\Lambda}(h,w,d)\triangleq
\big\{(x,\lambda)\in\calC\times\Lambda\mid x \mathrm{\ is}\ (h,w,d)\mathrm{-strongly~locally-optimal\ w.r.t.}\ \lambda\big\}.
\end{equation}

The following lemma states that if a codeword $x$ is strongly locally-optimal w.r.t. $\lambda$, then $x$ is locally-optimal w.r.t. $\lambda$.
\begin{lemma}\label{lemma:NLOinLO}
For every $\Lambda\subseteq\R^N$,
\[\NLO_{\calC,\Lambda}(h,w,d)\subseteq \LO_{\calC,\Lambda}(h,w,d).\]
\end{lemma}
\begin{proof}
We prove the contrapositive statement. Assume that $x$ is not $(h,w,d)$-locally optimal w.r.t. $\lambda$. By Proposition~\ref{proposition:LOsymmetry}, $0^N$ is not $(h,w,d)$-locally optimal w.r.t. $\lambda^0\triangleq (-1)^x\ast\lambda$. Hence, there exists a deviation $\beta=\pi_{G,\calT,w}\in\calB_d^{(w)}$ such that $\langle\lambda^0,\beta\rangle\leqslant0$. Let $\calT$ denote the $d$-tree that corresponds to the deviation $\beta$.

Denote by $(r)$ the root of $\calT$. By Lemma~\ref{lemma:trim}, the root $(r)$ has a child $q$ such that  $\langle\lambda^0,\pi_{G,\calT,w}\rangle \geqslant \langle\lambda^0,\pi_{G,\Trim(\calT,q),w}\rangle$.
Note that $\Trim(\calT,q)$ is a reduced $d$-tree rooted at $r$. Moreover, the corresponding reduced deviation
$\beta'=\pi_{G,\calT',w}$ satisfies $\langle \lambda^0,\beta'\rangle\leqslant\langle\lambda^0,\beta\rangle\leqslant0$.
We conclude by Proposition~\ref{proposition:LOsymmetry} that $x$ is not $(h,w,d)$-strong locally-optimal w.r.t. $\lambda$, as required.
\end{proof}

\medskip \noindent
Following Lemma~\ref{lemma:NLOinLO} and Theorem~\ref{thm:MLsufficient} we have the following corollary.
\begin{corollary}[strong local-optimality is sufficient for both ML and LP] \label{corr:ML-LPsufficient}
Let $\mathcal{C}(G)$ denote a Tanner code with minimum local distance $d^*$. Let $h\in\N_+$ and $w\in\R_+^h$. Let $\lambda\in\R^N$ denote the LLR vector received from the channel. If $x$ is an $(h,w,d)$-strong locally-optimal codeword w.r.t. $\lambda$ and some $2\leqslant d \leqslant d^{*}$, then (1)~$x$ is the unique maximum-likelihood codeword w.r.t. $\lambda$, and (2)~$x$ is the unique solution of LP-decoding given $\lambda$.
\end{corollary}

Consider a weight vector $\bar{w}\in\R^{k\cdot h}$, and let
$\bar{w} = \bar{w}^1\circ\bar{w}^2\circ\ldots\circ\bar{w}^k$ denote its decomposition to $k$ blocks $\bar{w}^i\in\R^h$. We say that $\bar{w}\in\R^{k\cdot h}$ is a \emph{$k$-legal extension of $w\in\R^h$} if there exists a vector $\alpha\in\R^k$ such that $\bar{w}^i=\alpha_i\cdot w$.
Note that if $\bar{w}\in\R^{k\cdot h}$ is geometric, then it is a $k$-legal extension of the first block $\bar{w}^1$ in its decomposition.

The following theorem derives a hierarchy on the height of reduced deviations of strong local-optimality characterization.
\begin{theorem}[$h$-Hierarchy of strong LO]\label{thm:h-hierarchy}
For every $\Lambda\subseteq\R^N$, if $\bar{w}\in\R^{k\cdot h}$ is a $k$-legal extension of $w\in\R^h$, then
\[\NLO_{\calC,\Lambda}(h,w,d)\subseteq \NLO_{\calC,\Lambda}(k\cdot h,\bar{w},d).\]
\end{theorem}

\begin{proof}
We prove the contrapositive statement. Assume that $x$ is not $(k\cdot h,\bar{w},d)$-strong locally-optimal w.r.t. $\lambda$. Proposition~\ref{proposition:isoLO} implies that $0^N$ is not $(k\cdot h,\bar{w},d)$-strong locally-optimal w.r.t. $\lambda^0\triangleq (-1)^x\ast\lambda$. Hence, there exists a reduced deviation $\beta=\pi_{G,\calT,\bar{w}}\in\calB_d^{(\bar{w})}$ such that $\langle\lambda^0,\beta\rangle\leqslant0$. Let $\calT$ denote the reduced $d$-tree that corresponds to the reduced deviation $\beta$.

Let $\{\calT_j\}$ denote a decomposition of $\calT$ to reduced $d$-trees of height $2h$ as shown in Figure~\ref{fig:reducedDecomposition}, where leaves of a subtree are the roots of other subtrees.
Let $p_j$ denote the root of a reduced $d$-tree $\calT_j$ in the decomposition of $\calT$.
For each subtree $\calT_j$ let $\ell(\calT_j)$ denote its ``level'', namely, $\ell(\calT_j) \triangleq \lfloor\abs{p_j}/h\rfloor$. Then,
\[\pi_{G,\calT,\bar{w}} = \sum_{\{\calT_j\}} \alpha_{\ell(T_j)}\cdot\pi_{G,\calT_j,w}.\]
\begin{figure}
  \begin{center}
 \includegraphics[width=0.45\textwidth]{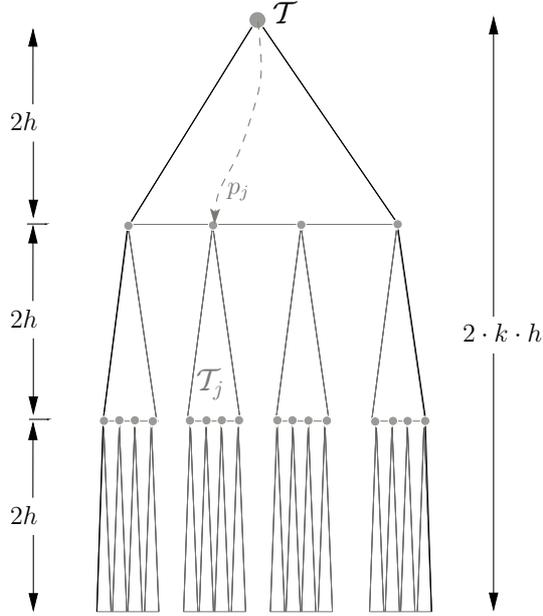}
 \caption{Decomposition of a reduced $d$-tree $\calT$ of height $2kh$ to a set of subtrees $\{\calT_j\}$ that are reduced $d$-trees of height $2h$.}
  \label{fig:reducedDecomposition}
  \end{center}
\end{figure}

Because $\langle\lambda^0,\beta\rangle\leqslant0$, we conclude by averaging that there exists at least one reduced $d$-tree $\calT^*\in\{\calT_j\}$ of height $2h$ such that $\langle\lambda^0,\pi_{G,\calT^*,w}\rangle\leqslant0$. Hence, $0^N$ is not $(h,w,d)$-strong locally-optimal w.r.t. $\lambda^0$. We apply Proposition~\ref{proposition:isoLO} again, and conclude that $x$ is not $(h,w,d)$-strong locally-optimal w.r.t. $\lambda$, as required.
\end{proof}

%%%%%%%%%%%%%%%%%%%%%%%%%%%%%%%%%%%%%%%%%%%%%%%%%%%%%%%%%%
\section{Numerical Results} \label{sec:numerical}
%%%%%%%%%%%%%%%%%%%%%%%%%%%%%%%%%%%%%%%%%%%%%%%%%%%%%%%%%%
We conducted simulations to demonstrate two phenomena.  First, we
checked the difference between strong local-optimality and local-optimality.
Second, we checked the effect of increasing the number of iterations
on successful decoding with ML-certificates (based on local-optimality).

We chose a $(3,6)$-regular LDPC code with blocklength $N=1008$ and
girth $g=6$~\cite{MacKay}.
For each $p\in\{0.04,0.05,0.06\}$, we randomly picked a set $\Lambda_p$ of $5000$
LLR vectors corresponding to the all zeros codeword with respect to a
BSC with crossover probability $p$.  We used unit
level weights, i.e., $w=1^h$, for the definition of local-optimality.

Let $\NLO_{0^N,\Lambda_p}(h,w,2)$ (resp.,
$\LO_{0^N,\Lambda_p}(h,w,2)$ ) denote the set of LLR vectors
$\lambda\in \Lambda_p$ such that $0^N$ is strongly locally-optimal
(resp., locally-optimal) w.r.t. $\lambda$.

Figure~\ref{fig:SLO} depicts cardinality of
$\NLO_{0^N,\Lambda_p}(h,w,2)$ and $\LO_{0^N,\Lambda_p}(h,w,2)$ as a
function of $h$, for three values of $p$.  The results suggest that,
in this setting, the sets $\NLO_{\{0^N\},\Lambda_p}(h,w,2)$ and
$\LO_{\{0^N\},\Lambda_p}(h,w,2)$ coincide as $h$ grows.  This suggests
that for finite-length codes and large height $h$, strong local-optimality is very close to local-optimality. For example, in our simulation for $p=0.04$ and $h=320$, $\lvert\LO_{\{0^N\},\Lambda_p}(h,w,2)\rvert=4868$ and $\lvert\NLO_{\{0^N\},\Lambda_p}(h,w,2)\rvert=4859$ (i.e., only $9$ LLRs out of $5000$ are in \LO\ but not in \NLO\ for height parameter $h=320$).

Iterative decoding is guaranteed to succeed after $h$ iteration if $(h,w,2)$-strongly locally-optimal w.r.t. $\lambda$. Hence, the results also suggest that the number of iterations needed to obtain reasonable decoding with ML-certificates is far greater than
the girth. Clearly, the ``tree property'' that DE analysis relies on
does not hold for so many iterations in finite-length codes.
Indeed, the simulated crossover
probabilities are in the ``waterfall'' region of the word error rate
curve with respect to \NWMS\ decoding. We are not aware of any analytic
explanation of the phenomena that iterative decoding of finite-length codes requires so many iterations in the ``waterfall'' region.

Another result of the simulation (for which we do not provide proof) is that
$\NLO_{0^N,\Lambda_p}(h,w,2)\subseteq
\NLO_{0^N,\Lambda_p}(h+1,w,2)$.  Namely, once a codeword is
strongly locally-optimal w.r.t. $\lambda$ with height $h$, then it is
also strongly locally-optimal for any height $h'>h$ (and not only multiples of $h$ as proved in Theorem~\ref{thm:h-hierarchy}). We point out that such a strengthening of the height hierarchy result is not true in general.

\begin{figure}
  \begin{center}
 \includegraphics[width=0.85\textwidth]{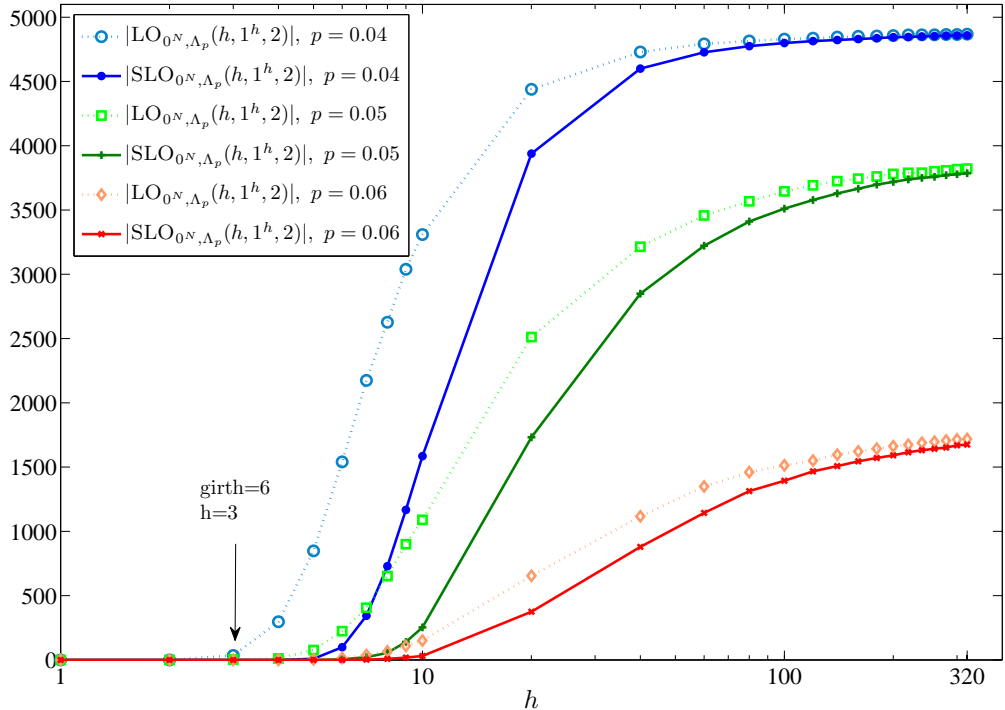}
 \caption{Growth of strong local-optimality and local-optimality as a function of the height $h$. $\lvert\Lambda_p\rvert=5000$ for $p\in\{0.04,0.05,0.06\}$.}
  \label{fig:SLO}
  \end{center}
\end{figure}

%%%%%%%%%%%%%%%%%%%%%%%%%%%%%%%%%%%%%%%%%%%%%%%%%%%%%%%%%%%%%%%%%%%%%%
\section{Discussion}
\label{sec:discussion}
%%%%%%%%%%%%%%%%%%%%%%%%%%%%%%%%%%%%%%%%%%%%%%%%%%%%%%%%%%%%%%%%%%%%%%
\paragraph{The degree hierarchy and probability of successful decoding of Tanner codes.}
The degree hierarchy supports the improvement in the lower bounds on the threshold value of the crossover probability $p$ of successful LP-decoding over a BSC$_p$ as a function of the degree parameter $d$ (see~\cite[Theorem~27]{EH11}). These lower bounds are proved by
analyzing the probability of a locally-optimal codeword as a function of $p$ and the degree parameter $d$.
 For example, consider any
$(2,16)$-regular Tanner code with minimum local-distance $d^*=4$ whose
Tanner graph has logarithmic girth in the blocklength.
The bounds in~\cite{EH11} imply  a lower bound on the
threshold of $p_0=0.019$ with respect to degree parameter $d=3$.
On the other hand, the lower bound on the
threshold increases to $p_0=0.044$ with respect to degree parameter $d=4$. However, note that the degree hierarchy holds for local-optimality with any height parameter $h$, while the probabilistic analysis in~\cite{EH11} restricts the parameter $h$ by a quarter of the girth of the Tanner graph.

\paragraph{The height hierarchy of strong local-optimality and iterative decoding}
The motivation for considering the height hierarchy comes from an
iterative message-passing algorithm (\NWMS) that is guaranteed to
successfully decode a locally-optimal codeword in $h$
iterations~\cite[Theorem~16]{EH11}.  Consider a Tanner code with
single parity-check local codes. Assume that $x$ is a codeword that is
strongly locally-optimal w.r.t. $\lambda$ for height parameter $h$.
Our results imply that:
\begin{inparaenum}[(i)]
\item $x$ is also strongly locally-optimal w.r.t. $\lambda$ for any
  height parameter $k\cdot h$ where $k\in\N_+$ (this is implied by the height hierarchy in
Theorem~\ref{thm:h-hierarchy}),
\item $x$ is also locally-optimal (this is implied by Lemma~\ref{lemma:NLOinLO}).
\end{inparaenum}
Therefore, we have that $x$ is also locally-optimal w.r.t. $\lambda$
for any height parameter $k\cdot h$ where $k\in\N_+$.
Thus \NWMS\ decoding is guaranteed to decode $x$ after $k\cdot h$
iterations~\cite[Theorem~16]{EH11}. This gives the following new
insight of convergence.  If a codeword $x$ is decoded after $h$
iterations and is certified to be strongly locally-optimal (and hence
ML-optimal), then $x$ is the outcome of \NWMS\ infinitely many times
(i.e., whenever the number of iterations is a multiple of $h$).

Richardson and Urbanke proved a monotonicity property w.r.t. iterations for belief propagation decoding of LDPC codes based on a tree-like setting and channel degradation~\cite[Lemma 4.107]{RU08}. Such a monotonicity property does not hold in general for suboptimal iterative decoders. In particular, the standard min-sum algorithm is not monotone for LDPC codes. The height hierarchy implies a monotonicity property w.r.t. iterations for \NWMS\ decoding with strong local-optimality certificates even without assuming the tree-like setting and channel degradation. That is, the performance of strongly locally-optimal \NWMS\ decoding of finite-length Tanner codes with SPC local codes does not degrade as the number of iterations increase, even beyond the girth of the Tanner graph. Proving an analogous non-probabilistic combinatorial height hierarchy for BP is an interesting open question.

%%%%%%%%%%%%%%%%%%%%%%%%%%%%%%%%%%%%%%%%%%%%%%%%%%%%%%%%%%%%%%%%%%%%%%
\section{Conclusion} \label{sec:conclusion}
%%%%%%%%%%%%%%%%%%%%%%%%%%%%%%%%%%%%%%%%%%%%%%%%%%%%%%%%%%%%%%%%%%%%%%
We present hierarchies of local-optimality with respect to two
parameters of the local-optimality characterization for Tanner
codes~\cite{EH11}.  One hierarchy is based on the local-code
node degrees in the deviations.  We prove containment, namely, the set
of locally-optimal codewords with respect to degree $d+1$ is a superset
of the set of locally-optimal codewords with respect to degree $d$.

The second hierarchy is based on the height of the deviations. We
prove that, for geometric level weights, a strongly locally-optimal
codeword is infinitely often strongly locally-optimal. In particular, a codeword that is decoded with a certificate using the
iterative decoder \NWMS\ after $h$ iterations is decoded with a
certificate after $k\cdot h$ iterations, for every integer $k$.

\appendix
\section{Proof of Lemma~\ref{lemma:trim}}\label{app:trimProof}
Let us first introduce the following averaging proposition.
\begin{proposition}\label{prop:avg}
Let $x_1,\ldots,x_k$ denote $k$ real numbers. Define $k_{\max}\triangleq\argmax_{1\leqslant i\leqslant k}\{x_i\}$, and
\[x_i' \triangleq \begin{cases}0 &\mathrm{if}~i=k_{\max},\\ \frac{k}{k-1}\cdot x_i&\mathrm{otherwise}.\end{cases}\]
Then, $\sum_{i=1}^k x_i \geqslant \sum_{i=1}^k x_i'$.
\end{proposition}
\begin{proof}
It holds that
\begin{align*}
\sum_{i=1}^k x_i' &= \sum_{i\neq k_{\max}} \frac{k}{k-1}\cdot x_i\\
&=  \sum_{i=1}^k x_i + \sum_{i\neq k_{\max}} \frac{1}{k-1}\cdot x_i -x_{k_{\max}}.
\end{align*}
Therefore, it is sufficient to show that
$x_{k_{\max}}\geqslant\sum_{i\neq k_{\max}} \frac{1}{k-1}\cdot x_i$.
The proposition follows because $x_{k_{\max}}$ is indeed greater or equal than the average of the other numbers.
\end{proof}

\begin{proof}[Proof of Lemma~\ref{lemma:trim}]
Consider a path $p\in\calT$, and let $p'$ denote a child of $p$ (i.e., $p'$ is an augmentation of $p$ by a single edge). We separate the inner products $\langle\lambda,\pi_{G,\calT,w}\rangle$ and $\langle\lambda,\pi_{G,\Trim(\calT,p'),w}\rangle$ to variable paths in $\hat\calV\setminus\hat V_p$ and in $\hat\calV\cap\hat V_p$ as follows.
\begin{equation}\label{eqn:trim1}
\langle\lambda,\pi_{G,\calT,w}\rangle=\underbrace{\sum_{q\in\hat\calV\setminus\hat{V}_p}\lambda_{t(q)}\cdot w_\calT(q)}_{(a)} + \underbrace{\sum_{q\in\hat\calV\cap\hat{V}_p}\lambda_{t(q)}\cdot w_\calT(q)}_{(b)}.
\end{equation}
\begin{equation}\label{eqn:trim2}
\langle\lambda,\pi_{G,\Trim(\calT,p'),w}\rangle=\underbrace{\sum_{q\in\hat\calV\setminus\hat{V}_p}\lambda_{t(q)}\cdot w_\calT(q)}_{(a')} + \underbrace{\sum_{q\in\hat\calV\cap\hat{V}_p}\lambda_{t(q)}\cdot w_\calT(q)}_{(b')}.
\end{equation}
It is sufficient to show: (i)~$\forall p'$ child of $p$: Term~(\ref{eqn:trim1}.a) $=$ Term~(\ref{eqn:trim2}.a'), and (ii)~$\exists p'$ child of $p$ s.t. Term~(\ref{eqn:trim1}.b) $\geqslant$ Term~(\ref{eqn:trim2}.b').

First we deal with the equality Term~(\ref{eqn:trim1}.a) $=$ Term~(\ref{eqn:trim2}.a'). Let $p'$ denote a child of $p$. For each $q\in\hat\calV\setminus\hat V_p$, it holds that $w_\calT(q)=w_{\Trim(\calT,p')}(q)$. Therefore,
\begin{equation}\label{eqn:trim3}
\sum_{q\in\hat\calV\setminus\hat{V}_p}\lambda_{t(q)}\cdot w_\calT(q) = \sum_{q\in\hat\calV\setminus\hat{V}_p}\lambda_{t(q)}\cdot w_{\Trim(\calT,p')}(q)
\end{equation}
Hence, Term~(\ref{eqn:trim1}.a) remains unchanged by trimming $\calT_{p'}$ from $\calT$ for every child $p'$ of $p$.

%It suffices to show that there exists a child $q_1$ of $p$ whose trimming does not increase term (b).

For a path $q$, let $\cost_\calT(\calT_q)\triangleq\sum_{q'\in\hat\calV_q}\lambda_{t(q')}w_\calT(q')$ denote the cost of $\calT_q$ with respect to $\calT$. Note that Term~(\ref{eqn:trim1}.b) equals $\cost_\calT(\calT_p)$.
We may reformulate Term~(\ref{eqn:trim1}.b) as follows:
\begin{equation}\label{eqn:trim4}
\cost_\calT(\calT_p)=\begin{cases}\lambda_{t(p)}w_\calT(p)+\sum_{ \{q\in\calN_\calT(p)\ :\ \lvert q\rvert = \lvert p\rvert + 1\}}\cost_\calT(\calT_q) &\mathrm{if}~t(p)\in\calV,\\
\sum_{ \{q\in\calN_\calT(p)\ :\ \lvert q\rvert = \lvert p\rvert + 1\}}\cost_\calT(\calT_q) &\mathrm{if}~t(p)\in\calJ.\end{cases}
\end{equation}
Consider two children $q_1$ and $q_2$ of $p$. By Definition~\ref{def:weightedSubtree}, for every variable path $q\in\calT_{q_2}$,
\begin{equation}
(\deg_{\calT}(p)-1)\cdot w_\calT(q) = (\deg_\calT(p)-2)\cdot w_{\Trim(\calT,q_1)}(q).
\end{equation}
Hence by summing over all the variable paths in $\calT_{q_2}$ we obtain
\begin{equation}
(\deg_{\calT}(p)-1)\cdot \cost_\calT(\calT_{q_2}) = (\deg_\calT(p)-2)\cdot \cost_{\Trim(\calT,q_1)}(\calT_{q_2}).
\end{equation}
Therefore,
\begin{equation}\label{eqn:trim5}
\frac{\cost_\calT(\calT_{q_2})}{\cost_{\Trim(\calT,q_1)}(\calT_{q_2})}=\frac{\deg_\calT(p)-2}{\deg_\calT(p)-1}\leqslant1.
\end{equation}

Let $q^{\max}$ denote a child of $p$, for which the subtree hanging from it has a maximum cost. Formally, $q^{\max}\triangleq\argmax\{\cost_\calT(\calT_q)\mid q\in\calN_\calT(p), \lvert q\rvert=\lvert p\rvert+1\}$.
We apply Proposition~\ref{prop:avg} as follows. Let $k=\deg_\calT(p)-1$, and let $x_i = \cost_\calT(\calT_{q_i})$ where $q_i$ denotes the $i$th child of $p$. Notice that by Equation~(\ref{eqn:trim5}), $x_i' = \cost_{\Trim(\calT,q^{\max})}(\calT_{q_i})$. It follows that
\begin{equation}\label{eqn:trim6}
\sum_{ \{q\in\calN_\calT(p)\ :\ \lvert q\rvert=\lvert p\rvert + 1\}}\cost_\calT(\calT_q) \geqslant \sum_{ \{q\in\calN_\calT(p)\ :\ \lvert q\rvert = \lvert p\rvert + 1\}\setminus\{q^{\max}\}}\cost_{\Trim(\calT,q^{\max})}(\calT_q).
\end{equation}

Because $\lambda_{t(p)}w_\calT(p)$ is unchanged by trimming a child of $p$, it follows from Equations~(\ref{eqn:trim4}) and~(\ref{eqn:trim6}) that
\begin{equation}\label{eqn:trim7} \cost_{\calT}(\calT_p)\geqslant\cost_{\Trim(\calT,q^{\max})}(\calT_p).
\end{equation}
Hence, we conclude that Term~(\ref{eqn:trim1}.b) $\geqslant$ Term~(\ref{eqn:trim2}.b') for $p'=q^{\max}$, and the lemma follows.
\end{proof}

%%%%%%%%%%%%%%%%%%%%%%%%%%%%%%%%%%%%%%%%%%%%%%%%%%%%%%%%%%%%%%%%%%%%%%%%
%\section*{Acknowledgement}
%%%%%%%%%%%%%%%%%%%%%%%%%%%%%%%%%%%%%%%%%%%%%%%%%%%%%%%%%%%%%%%%%%%%%%%%

%\bibliographystyle{plain}
%\small
%\clearpage
%\bibliographystyle{abbrv}
\bibliographystyle{alpha}
%\bibliography{ECCbib}

\newcommand{\etalchar}[1]{$^{#1}$}

\end{document}